\RequirePackage{fix-cm}
\documentclass[article]{svjour3}       % onecolumn (second format)

\def\makeheadbox{\relax}

\smartqed  % flush right qed marks, e.g. at end of proof
\usepackage{graphicx}
\usepackage{times}
\usepackage{soul}
\usepackage{url}
\usepackage[hidelinks]{hyperref}
\usepackage[utf8]{inputenc}
\usepackage[small]{caption}
\usepackage{graphicx}
\usepackage{amsmath}
\usepackage{amssymb}
\usepackage{array}
\usepackage{mathtools}
\usepackage{xspace}
\usepackage{booktabs}
\usepackage{adjustbox}
\usepackage{multirow,bigdelim}
\usepackage{tikz}
\usetikzlibrary{arrows}
\usepackage{caption}
\usepackage{float}
\usepackage{comment}
\usepackage[linesnumbered,noend]{algorithm2e}

\newcommand{\classFont}[1]{\protect\ensuremath{\mathsf{#1}}\xspace}

\newcommand{\sub}{\subseteq}
\def\boto{\mkern1.5mu\bot\mkern2.5mu}
\newcommand{\True}{\classFont{True_{\bot}}}
\newcommand{\NP}{\classFont{NP}}
\newcommand{\Wsym}{\classFont{W}}
\newcommand{\W}[1]{\classFont{W}{[#1]}}
\newcommand{\IAD}{\problemFont{ISD}}
\newcommand{\FPT}{\problemFont{FPT}}
\newcommand{\IIAD}{\problemFont{IISD}}
\newcommand{\IAcover}{\problemFont{IScover}}
\newcommand{\MINCOVER}{\classFont{MPCOVER}}
\newcommand{\CANDIDATES}{\classFont{CANDIDATES}}
\newcommand{\NEWCANDIDATES}{\classFont{NEWCANDIDATES}}
\newcommand{\N}{\mathbb{N}}
\newcommand{\tuple}[1]{\vec{#1}}
\newcommand{\problemFont}[1]{\protect\ensuremath{\textsf{#1}}}

\def\mystrut(#1,#2){\vrule height #1 depth #2 width 0pt}

\newcolumntype{C}[1]{%
   >{\mystrut(3ex,2ex)\centering}%
   p{#1}%
   <{}}

\begin{document}
	
	\title{An Algorithm for the Discovery of \\ Independence from Data}
	
	\titlerunning{Discovering Independence}        % if too long for running head
	
	\author{Miika Hannula     \and Bor-Kuan Song \and
		Sebastian Link%\footnote{Corresponding Author} %etc.
	}

	\institute{M. Hannula \at
		University of Helsinki \\
		Finland \\
		\email{miika.hannula@helsinki.fi}
		\and B.K. Song \at
		The University of Auckland\\
		New Zealand \\
		\email{bson010@aucklanduni.ac.nz}
		\and S. Link \at
		The University of Auckland\\
		New Zealand \\
		\email{s.link@auckland.ac.nz}
	}
	
	\date{\today}

\maketitle

\begin{abstract}
For years, independence has been considered as an important concept in many disciplines. Nevertheless, we present the first research that investigates the discovery problem of independence in data. In its arguably simplest form, independence is a statement between two sets of columns expressing that for every two rows in a table there is also a row in the table that coincides with the first row on the first set of columns and with the second row on the second set of columns. We show that the problem of deciding whether there is an independence statement that holds on a given table is not only \NP-complete but \W{3}-complete in its arguably most natural parameter, namely its arity. We establish the first algorithm to discover all independence statement that hold on a given table. We illustrate in experiments with benchmark data that our algorithm performs well within the limits established by our hardness results. In practice, it is often useful to determine the ratio with which independence statements hold on a given table. For that purpose, we show that our treatment of independence and the design of our algorithm enables us to extend our findings to approximate independence. In our final experiments, we provide some insight into the trade-off between run time and the approximation ratio. Naturally, the smaller the ratio, the more approximate independence statements hold, and the more time it takes to discover all of them. While this research establishes first insight into the computational properties of discovering independence from data, we hope to initiate research into more sophisticated notions of independence, including embedded multivalued dependencies, as well as their context-specific and probabilistic variants.
\end{abstract}


\section{Introduction}

Independence is a fundamental concept in areas as diverse as artificial intelligence, databases, probability theory, social choice theory, and statistics \cite{dawid:1979,halpern:2005,pearl:1988}. Indeed, valid independencies often facilitate efficient computations, effective analytics, and knowledge discovery tasks. This holds to an extent where independence is often simply assumed to hold. Classical examples include the independence assumptions in databases \cite{DBLP:journals/pvldb/ChaudhuriNR09,DBLP:conf/vldb/PoosalaI97}, the independence assumption in artificial intelligence \cite{Koller:2009:PGM:1795555,DBLP:journals/ai/PednaultZM81}, or the independence assumption in information retrieval \cite{DBLP:conf/ecml/Lewis98}. However, such assumptions are often wrong and lead to incorrect results. Obviously, we would like to move away from such guesswork, and have reliable information on what independence structures do hold in our data.

We therefore ask the question how difficult it is to discover independence. As a plethora of notions for independence exist and the foundations of many disciplines deeply depend on these concepts, it is surprising that their discovery problem has not received much attention. In fact, data profiling has evolved to becoming an important area of database research and practice. The discovery problem has been studied in depth for many important classes of data dependencies, such as unique column combinations, foreign keys, or functional dependencies \cite{Heise:2013,Papenbrock:2015,Zhang:2010}. However, data dependencies related to the concepts of independence have not received much attention. We thus pick the arguably simplest notion of independence to initiate more research into this area. Our notion of independence is simply stated in the context of a given finite set $r$ of tuples over a finite set $R$ of attributes. We denote the independence statement (IS) between two attribute subsets $X,Y$ of $R$ by $X\boto Y$. We say that the IS $X\boto Y$ \emph{holds on} $r$ if and only if for every pair of tuples $t_1,t_2\in r$ there is some tuple $t\in r$ such that $t(X)=t_1(X)$ and $t(Y)=t_2(Y)$. In other words, $X\boto Y$ holds on $r$ if and only if the projection $r(XY)$ of $r$ onto $X\cup Y$ is the Cartesian product $r(X)\times r(Y)$ of its projections $r(X)$ onto $X$ and $r(Y)$ onto $Y$. \emph{Discovery of independence} now refers to the problem of computing all ISs over $R$ that hold on a given relation $r$ over $R$.

As an illustrative real-world like example let us project the tuples 1, 2, 3, and 7110 of the discovery benchmark data set \emph{adult} onto its columns 1, 4, 8, 9, and 10, as shown in Table~\ref{t:ex}. Here, the IS \[\emph{education}\boto\emph{relationship}\] holds but not the IS \[\emph{education}\boto\emph{relationship},\emph{sex}\] since there is no tuple with projection (bachelors, not-in-family, female). Also, the IS \[\emph{race}\boto\emph{race}\] holds, expressing that the column \emph{race} has at exactly one value (that is, constant).

\begin{table}
\centering
\caption{Relation \emph{adult\_sub} of \emph{adult} data set\label{t:ex}}
\begin{tabular}{ccccc}\hline
\emph{{\bf a}ge} & \emph{{\bf e}ducation} & \emph{{\bf re}lationship} & \emph{{\bf ra}ce} & \emph{{\bf s}ex} \\ \hline
39 & bachelors & not-in-family & white & male \\
50 & bachelors & husband & white & male \\
38 & hs-grad & not-in-family & white & male \\
34 & hs-grad & husband & white & female \\ \hline
\end{tabular}
\end{table}

The contributions and organization of our article are summarized as follows. After stating the discovery problem formally in Section~\ref{s:problem}, we establish the computational complexity for a natural decision variant of the discovery problem in Section~\ref{s:complexity}. Indeed, the decision variant is \NP-complete and $\W{3}$-complete in the arity of the input. This is only the second natural problem to be $\W{3}$-complete, to the best of our current knowledge. In Section~\ref{s:alg} we propose the first algorithm for the discovery problem. It is worst-case exponential in the number of given columns, but can handle many real-world data sets efficiently. This is illustrated in Section~\ref{s:experiments}, where we apply the algorithm to several real-world discovery benchmark data sets. In Section~\ref{sec:approximate} we show that our discovery algorithm can be easily augmented to discover approximate independence statements. In the same section we illustrate on our benchmark data the tradeoff between runtime efficiency for discovering approximate ISs and the minimum bounds with which these approximate ISs must hold on the data set. We report on related work in Section~\ref{sec:related}, and give a detailed plan of future work in Section~\ref{sec:future}. We conclude in Section~\ref{s:conclusion}.

We envision at least four areas of impact: 1) initiating research on the discovery problem for notions of independence in different disciplines such as databases and artificial intelligence, 2) laying the foundations for the discovery of more sophisticated notions of independence, such as context-specific or probabilistic independence, 3) replacing independence assumptions by knowledge about which independence statements hold or to which degree they hold, in order to facilitate more accurate and more efficient computations, for example in cardinality estimations for query planning, and 4) augmenting current data profiling tools by classes of independence statements.

\section{Problem Statement}\label{s:problem}

Let $\mathfrak{A}=\{A_1,A_2,\ldots\}$ be a (countably) infinite set of symbols, called \emph{attributes}. A \emph{schema} is a finite set $R=\{A_1,\ldots,A_n\}$ of attributes from $\mathfrak{A}$. Each attribute $A$ of a schema is associated with a domain $\textit{dom}(A)$ which represents the set of values that can occur in column $A$. A \emph{tuple} over $R$ is a function $t:R\rightarrow\bigcup_{A\in R}\textit{dom}(A)$ with $t(A)\in\textit{dom}(A)$ for all $A\in R$. For $X\subseteq R$ let $t(X)$ denote the restriction of the tuple $t$ over $R$ on $X$. A \emph{relation} $r$ over $R$ is a finite set of tuples over $R$. Let $r(X)=\{t(X)\mid t\in r\}$ denote the \emph{projection} of the relation $r$ over $R$ on $X\subseteq R$. For attribute sets $X$ and $Y$ we often write $XY$ for their set union $X\cup Y$.

Intuitively, an attribute set $X$ is independent of an attribute set $Y$, if $X$-values occur independently of $Y$-values. That is, the independence holds on a relation, if every $X$-value that occurs in the relation occurs together with every $Y$-value that occurs in the relation. Therefore, we arrive at the following concept of independence. An \emph{independence statement} (IS) over relation schema $R$ is an expression $X\boto Y$ where $X$ and $Y$ are subsets of $R$. A relation $r$ over $R$ is said to \emph{satisfy} the IS $X\boto Y$ over $R$ if and only if for all $t_1,t_2\in r$ there is some $t\in r$ such that $t(X)=t_1(X)$ and $t(Y)=t_2(Y)$. If $r$ does not satisfy $X\boto Y$, then we also say that $r$ \emph{violates} $X\boto Y$. Alternatively, we also say that $X\boto Y$ holds on $r$ or does not hold on $r$, respectively. In different terms, given disjoint $X$ and $Y$ $r$ satisfies $X\boto Y$ if and only if $r(XY)=r(X)\times r(Y)$. Hence, for a relation $r$ which satisfies the IS $X\boto Y$, the projection $r(XY)$ is the lossless Cartesian product of the projections $r(X)$ and $r(Y)$.

In what follows we use $\Sigma$ to denote sets of ISs, usually over some fixed schema $R$. We say that $\Sigma' \sub \Sigma$ is a \emph{quasi-cover} of $\Sigma$ if for all $X\boto Y$ there is some $X'\boto Y'\in \Sigma'$ such that $X\sub X'\wedge Y\sub Y'$ or $Y\sub X'\wedge X\sub Y'$. $\Sigma'$ is called a \emph{cover} of $\Sigma$ if $\Sigma'\models \Sigma$. A (quasi-)cover is called \emph{minimal} if none of its proper subsets is a (quasi-)cover. Minimal quasi-covers are unique up to the symmetry: $X\boto Y$ holds if and only if $Y\boto X$ holds. Let $\True(r)$ be the set of all ISs satisfied by $r$. The discovery problem is to compute a minimal cover of $\True(r)$ as stated in Table~\ref{discproblem}.

\begin{table}
\begin{center}
\caption{Independence discovery \label{discproblem}}
\begin{tabular}{rl}
\toprule
\textbf{Problem:} & \IAcover\\\hline
\textbf{Input:} &A relation $r$\\
\textbf{Output:} &A minimal cover of $\True(r)$\\ \hline
\end{tabular}
\end{center}
\end{table}

\section{Likely Intractability}\label{s:complexity}

Before attempting the design of any efficient algorithms for the discovery problem of independence statements it is wise
to look at the computational complexity of the underlying problem. Indeed, we can establish the \NP- and $\W{3}$-completeness,
which provides great insight into the limitations and opportunities for efficient solutions in practice. In fact, the
computational difficulty of the discovery problem for independence statements is therefore on par with that of inclusion
dependencies.

 \begin{table}
\begin{center}
\caption{Decision variant of independence discovery \label{decproblem}}
\begin{tabular}{rl}\hline
\textbf{Problem:} & $\IAD$\\ \hline
\textbf{Input:} &A relation $r$ and a natural number $k$ \\
\textbf{Output:} &Yes iff $r$ satisfies a non-trivial \\
                 & $k$-ary independence statement. \\ \hline
\end{tabular}
\end{center}
\end{table}

\subsection{Independence Discovery is $\NP$-complete}

First we show that independence discovery is $\NP$-complete. The decision version of the discovery problem is given in terms of the arity of the independence statement.
\begin{definition}
The arity of an independence statement $X\boto Y$ is defined as the number of its distinct attributes $k=|X\cup Y|$. Such an independence statement is then called $k$-ary.
\end{definition}
Note that a relation $r$ satisfies some $k$-ary independence statement if and only if it satisfies some $l$-ary independence statement for $l\geq k$. The decision version of the problem, hereafter referred to as \IAD, is now given in Table \ref{decproblem}. We call an independence statement $X\boto Y$ \emph{non-trivial} if both $X$ and $Y$ are non-empty.

For $\NP$-hardness of \IAD, we reduce from the node biclique problem which is the problem of finding a maximal subset of nodes that induce a complete biclique. Given a graph $G=(V,E)$ and a natural number $k$, the \emph{node biclique problem} is to decide whether there are two disjoint non-empty sets of nodes $V_1,V_2$ such that  $|V_1|+|V_2|=k$; $u\in V_1,v\in V_2$ implies $\{u,v\}\in E$; and $u,v\in V_i$ implies $\{u,v\}\not\in E$ for $i=1,2$. That this problem is $\NP$-hard follows from arguments in \cite{Yannakakis81a}. A simple way to show $\NP$-hardness is to reduce from the \emph{independent set} problem (the following reduction is from \cite{Hochbaum98}). Given a graph $G=(V,E)$ construct a graph $G^2$ by taking a distinct copy $V'$ of the node set $V$ and $E'$ of the edge set $E$, and by adding an edge between each node in $V$ and each node in $V'$. Then  a biclique in $G^2$ is any pair of independent sets from $V$ and $V'$, which means that maximizing the size of the biclique in $G^2$ maximizes the size of the independent set in $G$. Note that the node biclique problem has several variants, and the exact way to formulate the problem may have an effect on its complexity. For instance, if $V_1$ and $V_2$ are allowed to have internal edges or if $G$ is bipartite, then the problem can be decided in polynomial time \cite{GareyJ79,Hochbaum98}.

\begin{theorem}
$\IAD$ is $\NP$-complete.
\end{theorem}

\begin{proof}
The membership in $\NP$ is easy to verify: guess a $k$-ary IS and verify in polynomial time whether it is satisfied in the given relation. For $\NP$-hardness we reduce from the node biclique problem. Given a graph $G=(V,E)$ we define a relation $r$ with the node set $V$ as its relation schema. The relation $r$ is constructed as follows (see Fig. \ref{ex}):
% \textbf{Balanced Complete Bipartite Subgraph} in which one is given a bipartite graph $G=(V,E)$ and a positive integer  $n\leq |V|$, and the problem is to decide whether $G$ contains a complete bipartite subgraph  $K_{n,n}$, i.e., whether there exist two disjoint  $V_0,V_1\sub V$ such that $|V_0|=|V_1|=n$ and such that $u\in V_0$ and $v\in V_1$ implies $(v_0,v_1)\in E$. By \cite{GareyJ79} this problem is $\NP$-complete, and it is easy to see that it remains such for graphs with no isolated nodes. In what follows, we construct a relation $r$ from $G=(V,E)$ such that $G$ contains $K_{n,n}$ iff $r$ satisfies some $X\boto Y$ with $|X|=|Y|=n$. We let $r$ be a relation over $R$ where $R$ is the vertex set $V$. Tuples of $r$ are constructed as follows:
\begin{itemize}
\item for each edge $(u,v)$ add a tuple that maps $u$ and $v$ to $1$ and every other attribute to $0$;
\item for each node $u$ add a tuple that maps $u$ to $1$ and every other attribute to $0$; and
\item add one tuple that maps all attributes to $0$.
\end{itemize}
%Note that since $G$ has no isolated nodes, each column in $r$ lists values $0$ and $1$.
It suffices to show that two disjoint node sets $V_1$ and $V_2$ induce a biclique in $G$ iff $V_1 \boto V_2$ is an independence statement of $r$.
%$K_{n,n}$ iff  $r$ satisfies some non-trivial $X\boto Y$ with $|X|=|Y|=n$.
The only-if direction is straightforward by the construction, so let us consider only the if direction. Assume that $r\models V_1\boto V_2$. First notice that $V_1$ and $V_2$ must be disjoint since no column in $r$ has a constant value. Let $v_1\in V_1$ and $v_2\in V_2$. Then we find $s_1$ ($s_2$, resp.) from $r$ mapping $v_1$  ($v_2$, resp.) to $1$. By assumption we find $s_3$ from $r$ mapping both $v_1$ and $v_2$ to $1$, which means that $v_1$ and $v_2$ are joined by an edge. Assume then to the contrary that two nodes $v_1,v_2$ from $ V_1$ are joined by an edge. Then we find $s_1$ from $r$ mapping both $v_1$ and $v_2$ to $1$. Since $V_2$ is non-empty, we find $s_2$ mapping some $v_3$ from $V_2$ to $1$. Then by assumption  there is a third tuple $s_3$ in $r$ mapping all $v_1,v_2,v_3$ to $1$. This contradicts the construction of $r$. By symmetry the same argument applies for $V_2$. Thus, we conclude that $V_1$ and $V_2$ induce a biclique.
 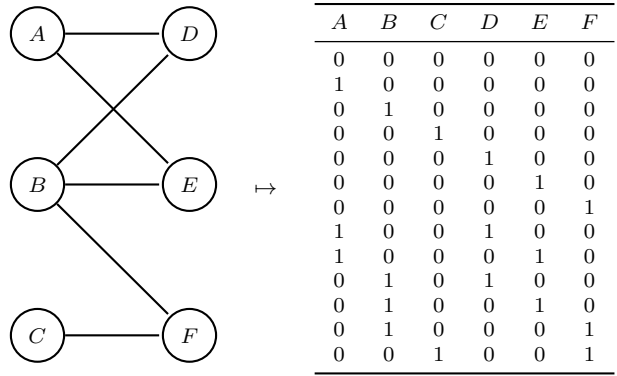
\begin{figure}
\begin{center}
\begin{tabular}{ccc}
\adjustbox{valign=m}{
\begin{tikzpicture}[->,>=stealth',shorten >=1pt,auto,node distance=2cm,
        thick,main node/.style={circle,draw,minimum size=.7cm,inner sep=0pt]}]

    \node[main node] (1) {$A$};
    \node[main node] (2) [right of=1]  {$D$};
    \node[main node] (3) [below of=1] {$B$};
        \node[main node] (4) [below of=2] {$E$};
        \node[main node] (5) [below of=3] {$C$};
        \node[main node] (6) [below of=4] {$F$};

    \path[-]
    (1) edge node {} (2)
        edge node {} (4)
    (3) edge node {} (2)
        edge node {} (4)
        edge node {} (6)
    (5) edge node {} (6);
\end{tikzpicture}
}
&
$\mapsto$
\vspace{0cm}
&
\begin{tabular}{cccccc}
\toprule
$A$&$B$&$C$&$D$&$E$&$F$\\\midrule
$0$&$0$&$0$&$0$&$0$&$0$\\
$1$&$0$&$0$&$0$&$0$&$0$\\
$0$&$1$&$0$&$0$&$0$&$0$\\
$0$&$0$&$1$&$0$&$0$&$0$\\
$0$&$0$&$0$&$1$&$0$&$0$\\
$0$&$0$&$0$&$0$&$1$&$0$\\
$0$&$0$&$0$&$0$&$0$&$1$\\
$1$&$0$&$0$&$1$&$0$&$0$\\
$1$&$0$&$0$&$0$&$1$&$0$\\
$0$&$1$&$0$&$1$&$0$&$0$\\
$0$&$1$&$0$&$0$&$1$&$0$\\
$0$&$1$&$0$&$0$&$0$&$1$\\
$0$&$0$&$1$&$0$&$0$&$1$\\\bottomrule
\end{tabular}
\end{tabular}
\end{center}\caption{Example reduction\label{ex}}
\end{figure}
\end{proof}

\subsection{Independence Discovery is $\W{3}$-complete}

Next we turn to the parameterized complexity of independence discovery. The parameterized complexity for the discovery of functional and inclusion dependencies was recently studied in \cite{Blasius0S16}. Parameterized on the arity, that is, the size of $X$ for both a functional dependency (FD) $X\to A$ and an inclusion dependency (IND) $r[X]\sub r'[Y]$, it was shown that the discovery problems are complete for the second and third levels of the $\Wsym$ hierarchy, respectively. %Indeed,
The case for inclusion dependencies is particularly interesting as many natural fixed-parameter problems usually belong to either $\W{1}$ or $\W{2}$. %, while the only natural problem know to date to be complete for $\W{3}$ is the discovery of inclusion dependencies.
 We show here that independence discovery is also $\W{3}$-complete in the arity of the independence statement. The \emph{arity} of the independence statement $X \boto Y$ is given as the size of the union $XY$.

A \emph{parameterized problem} is a language $L\sub \Sigma^* \times \N$, where $\Sigma$ is some finite alphabet.  The second component $k \in \N$ is called the \emph{parameter} of the problem.
The problem $L$ is called \emph{fixed-parameter tractable} (\FPT)  if $(x,k)\in L$ can be recognized in time $\mathcal{O}(f(k)p(|x|))$ where $p$ is a polynomial, and $f$ any computable function that depends only on $k$. Let $L$ and $L'$ be two parameterized  problems. An \emph{\FPT-reduction} from $L$ to $L$ is an $\FPT$ computable function that maps an instance $(x,k)$ of $L$ to an equivalent instance $(x',k')$ of $L'$ where the parameter $k'$ depends only on the parameter $k$. We write $L\leq_{\FPT} L'$ if there is an \FPT-reduction from $L$ to $L'$.

The relative hardness of a fixed-parameter intractable problem can be measured using the $W$ hierarchy. Instead of giving the standard definition via circuits, we here employ weighted satisfiability of propositional formulae \cite{DowneyF95,DowneyF95b}. A satisfying assignment of a propositional formula $\phi$ is said to have a \emph{hamming weight} $h$ if it sets exactly $h$  variables to true. A formula $\phi$ is $t$-normalized if it is a conjunction of disjunctions of conjunctions (etc.) of literals with $t-1$ alternations between conjunctions and disjunctions. For instance, CNF formulae are examples of $2$-normalized formulae. \emph{Weighted $t$-normalized satisfiability} is the problem to decide,  given a $t$-normalized formula $\phi$ and a parameter $k$, whether $\phi$  has a satisfying truth assignment with hamming weight $k$.
A parameterized problem $L$ is said to be in $\W{t}$ if $L \leq_{\FPT}  \W{t}$. These classes form a hierarchy $\FPT = \W{0} \sub \W{1} \sub \W{2} \sub \ldots $, and none of the inclusion are known to be strict.

Many natural fixed-parameter problems belong to the classes $\W{1}$ or $\W{2}$. %while the only natural problem know to date to be complete for $\W{3}$ is the discovery of inclusion dependencies.
 In what follows, we will show that independence discovery parameterized in the arity is complete for $\W{3}$.  To this end, it suffices relate to \emph{antimonotone} propositional formulae, that are, propositional formulae in which only negative literals may appear. \emph{Weighted antimonotone $3$-normalized satisfiability} (\problemFont{WA3NS}) is the problem to decide, given an antimonotone $3$-normalized formula $\phi$ and a parameter $k$, whether $\phi$ is a satisfying assignment with hamming weight $k$. By \cite{DowneyFellows1992b,DowneyF95,DowneyF95b} this problem remains $\W{3}$-complete in the parameter $k$. %The following presentation follows ideas from \cite{Blasius0S16} to a large extent.

\begin{lemma}\label{inW3}
$\IAD$  is in \W{3}.
\end{lemma}

\begin{proof}
Starting from a relation $r$ over a schema $R$ we build an antimonotone $3$-normalized formula $\phi$ such that $r$ satisfies a $k$-ary independence statement iff $\phi$ has a weight $k$ satisfying assignment. For this we introduce propositional variables $x_{i,P}$ where $i$ ranges over  $1, \ldots ,|R|$ and $P$  ranges over $L,C,R$.  A $k$-ary independence statement 
\[A_1 \ldots A_{l+m} \boto A_l \ldots A_{l+m+n}\] is now represented by setting $x_{h,L}$, $x_{i,C}$, $x_{j,R}$ to true for $h=1, \ldots ,l-1$, $i=l, \ldots ,l+m$, and $j=l+m+1, \ldots ,l+m+n$.  We then define
\[\phi_0= \bigwedge_{1 \leq i \leq |R|} (\neg x_{i,L} \vee \neg x_{i,C}) \wedge (\neg x_{i,L} \vee \neg x_{i,R})  \wedge (\neg x_{i,C} \vee \neg x_{i,R}).\]
Notice that any weight $k$ satisfying assignment of $\phi_0$ identifies a partition of a $k$-ary independence statement $X\boto Y$ to sets $X\setminus Y$, $X\cap Y$,  $Y\setminus X$.
% now describes a partial bijective mapping $f$ from $|R|$ into $|R|$ that identifies an $n$-ary IS $A_{i_1} \ldots A_{i_n}\boto A_{f(i_1)}, \ldots A_{f(i_n)}$ for $\Dom(f)=\{i_1, \ldots i_n\}$.
The satisfaction of this atom in $r$ is now encoded by introducing a fresh antimonotone $3$-normalized constraint. Let be $t_0,t_1,t_2$ be tuples from $r$. We say that a variable $x_{i,P}$ is \emph{forbidden} for $(t_0,t_1,t_2)$ if
\begin{itemize}
\item $t_0(A_i)\neq t_2(A_i)$ and $P=L$,
\item $t_1(A_i)\neq t_2(A_i)$ and $P=R$,
\item $t_0(A_i)\neq t_1(A_i)$ and $P=C$.
 \end{itemize}
We then let $F_{h,i,j}$ be the set of all forbidden variables for $(t_h,t_i,t_j)$, and set
\[\phi_{h,i,j}=\bigwedge_{x\in F_{h,i,j}} \neg x.\]
The independence statement $A_1 \ldots A_{l+m} \boto A_l \ldots A_{l+m+n}$ is satisfied by $r$ iff for all $t_h,t_i$ from $r$ there is $t_j$ from $r$ such that
\[x_{1,L},\ldots ,x_{l-1,L},x_{l,C},\ldots ,x_{l+m,C} ,x_{l+m+1,R},\ldots ,x_{l+m+n,R}\]
 are not forbidden for $(t_h,t_i,t_j)$. Therefore, letting
\[\phi_1 = \bigwedge_{1 \leq h < i \leq |R|} \bigvee_{1\leq j\leq |R|} \phi_{h,i,j}\]
we obtain that $r$ satisfies a $k$-ary independence statement iff $\phi_0\wedge \phi_1$ has a weight $k$ satisfying assignment. Since this formula is  $3$-normalized and antimonotone, and polynomial-time constructible from $r$, the claim of the lemma follows.
\end{proof}

Next we show that $\IAD$ is hard for $\W{3}$ by reducing from \problemFont{WA3NS}. To this end, it suffices to relate to a simplified version of the independence discovery problem. For a relation schema $R$ we call $R\cup\{i\}$  where $i$ is a special index column an \emph{index schema}. A relation over $R\cup\{i\}$ is called an \emph{index relation}.
The \emph{indexed independence statement discovery} ($\IIAD$) is given as an index relation $r$ over $R\cup\{i\}$ and a parameter $k$, and the problem is to decide whether $r$ satisfies an IS $X\boto i$ where $X$ is of size $k$.  This variant can be reduced to the general problem as follows. Given an index relation $r$ over $\{A_1, \ldots ,A_n,i\}$, and given two values $0$ and $1$ not appearing anywhere in $r$, increment $r$ with tuples $(0,\ldots ,0,a)$ and $(1,\ldots ,1,a)$ for all values $a$ appearing in the index column $i$. Then the incremented relation satisfies some $(k+1)$-ary independence statement iff the initial relation satisfies $X\boto i$ where $X$ is of size $k$.

\begin{lemma}\label{W3hard}
$\IIAD \leq_{\FPT} \IAD$.
\end{lemma}

\begin{proof}
Let $r$ over $R\cup\{i\}$ be an instance of \problemFont{Index IS discovery}. We construct from $r$  an equivalent instance $r'$ of \problemFont{IA discovery}, defined over a relation schema $R'= R\cup I$ where $I= \{i_A\mid A\in R\}$. To this end, we first let $r_0$ be the relation obtained from $r$ by replacing each tuple $t$ with a tuple $t'$ that extends $t[R]$ with $i_A\mapsto t(i)$ for $A\in R$. Then  let $t_0$ and $t_1$ be constant tuples over $R'$ that map all attributes to $0$ and $1$, respectively,  where $0$ and $1$ are two distinct values not appearing in $r$. We define $r_1=r_0\cup \{t_0,t_1\}$ furthermore $r'= r_0 \cup (  (r_1[R] \times r_1[I])\setminus (r_0[R]\times r_0[I]))$. We claim that $r$ satisfies an independence statement $X\boto i$ for $|X|=n$ iff $r'$ satisfies an $n$-ary independence statement.

For the only-if direction it suffices to note that, given an IS $X\boto i$ in $r$ where $|X|=n$, then $X\boto I$ is an IS in $r'$ by construction. For the if direction assume that $X\boto Y$ is an $n$-ary IS in $r'$. By construction $X\sub R$ and $Y\sub I$ (or vice versa), and hence $X\boto i$ is an IS in $r$. Since $|X|\geq n$ the claim follows.
\end{proof}

It suffices to show that \problemFont{WA3NS} reduces to $\IIAD$. To this end, we first prove the following lemma that will be applied to the outermost conjunctions of $3$-normalized  formulae.  Notice that
each propositional formula $\phi$ over $n$ variables gives rise to a truth assignment $f_{\phi}\colon\{0,1\}^n\to \{0,1\}$ defined in the obvious way. Similarly, following \cite{Blasius0S16} we define for each index relation $r$ over $R\cup\{i\}$ an \emph{indicator function} $f_r\colon\{0,1\}^{|R|}\to \{0,1\}$ so that it maps the characteristic function of  each $X\sub R$ to $1$ iff $r$ satisfies $X\boto i$. Lemma \ref{help1} is now analogous to its inclusion dependency discovery counterpart in \cite{Blasius0S16}.

\begin{lemma}\label{help1}
Let  $r_1, \ldots ,r_n$ be index relations over a shared index schema $R\cup\{i\}$. Then there is a relation $r$ over $R\cup\{i\}$ with size cubic in $|\bigcup_{j=1}^n r_j|$ and such that  $f_r= \bigwedge_{j=1}^n f_{r_j}$.
\end{lemma}

\begin{proof}
We may assume that distinct $r_i$ and $r_j$ do not have any values or indices in common. It suffices to define
\[r=\bigcup_{j=1}^n r_j\cup \bigcup_{j=1}^n  (r_j[R]\times \bigcup_{\substack{k=1\\k\neq j}}^n r_k[i]) .\]
\end{proof}

Given the previous lemma, it remains to construct a reduction from antimonotone DNF formulae to equivalent index relations. The proof of the lemma follows that of an analogous lemma in \cite{Blasius0S16}.
\begin{lemma}\label{help2}
$\phi$ be a antimonotone formula in DNF. Then there is an index relation $r$ with size cubic in $|\phi|$  and such that $f_{\phi}=f_r$.
\end{lemma}
\begin{proof}
 An antimonotone DNF formula $\phi$ is of the form
\[\bigvee^m_{j=1}  \psi_j \]
where each $\psi_j$ is a conjunction of negative literals. Assuming that the variables of $\phi$ are $X= \{x_1, \ldots ,x_n\}$, we define an index relation over $R=X\cup\{i\}$. This relation $r$ is defined as the union of two relations $r_0$ and $r_1$.

First, we define $r_0:=\{t_1, \ldots ,t_m\}$ where $t_j$ maps everything to $j$, except that variables $x_j$ that do not appear in $\psi_j$ are mapped to $0$. For instance, $\psi_1$ of the form $(\neg x_1\wedge \neg x_3\wedge \neg x_5)$ gives rise to a tuple $(1,0,1,0,1;1)$ where the last number denotes the index value (see Fig. \ref{ex1}).

Second,  we define $r_1 := r_X\times r_i$ where $r_X$ and $r_i$ have schemata $X$ and $\{i\}$. Let ``$-$'' be a value that does not appear in $r_0$. The relation $r_X$ is  given by generating $m$ copies of $r_0[X]$. In the $l$th copy we set each value of $x_j$ to this new value ``$-$'' whenever  $x_j$ appears in $\psi_l$. Then $r_X$ is obtained by taking the union of all these copies. Finally, we define $r_i$ over $i$ as $\{1, \ldots ,m\}$.

\begin{figure}[ht]\label{fig:example for extension}
\begin{center}
\begin{tabular}{cl}
$\phi$&$=(c_1\wedge c_2\wedge c_3)$\\
$c_1$&$=(\neg x_1\wedge \neg x_3\wedge \neg x_5)$\\
$c_2$&$=(\neg x_2\wedge \neg x_4\wedge \neg x_5)$\\
$c_3$&$=(\neg x_3\wedge \neg x_4)$
\end{tabular}
\qquad
\begin{tabular}{ccccc|clr}
			     $x_1$ &      $x_2$& $x_3$&$x_4$&$x_5$&$i$      \\\cline{1-6}
	  1 & 	0 & 1 & 0 & 1 	& 1 &\rdelim\}{3}{3mm}[$r_0$] 	 \\
	  0& 	2& 0 & 2 & 2 	 & 2	 \\
		  0& 	0 & 3 & 3 & 0 & 3		 \\\cline{1-6}

	  -& 	0 & - & 0 & - 	& $\{1,2,3\}$&	\rdelim\}{9}{3mm}[$r_1$]  \\
	  -& 	2& - & 2 & - 	 & $\{1,2,3\}$	 \\
		  -& 	0 & - & 3 & - & $\{1,2,3\}$		 \\

	  1 & 	- & 1 & - & - 	& $\{1,2,3\}$	 \\
	  0& 	-& 0 & - & - 	 & $\{1,2,3\}$	 \\
		  0& 	- & 3 & - & -& $\{1,2,3\}$		 \\

	  1 & 	0 & - & - & 1 	& $\{1,2,3\}$	 \\
	  0& 	2& -& - & 2 	 & $\{1,2,3\}$	 \\
		  0& 	0 & - & - & 0 & $\{1,2,3\}$		 \\
\end{tabular}

\end{center}
\caption{An example of a DNF formula $\phi$ and a relation $r=r_0\cup r_1$\label{ex1}}
\end{figure}

We claim that $f_{\phi}(Y)=f_r(Y)$ for any binary sequence $Y$ of length $|X|$. Assume first that $f_{\phi}(Y)=1$.  Then $f_{\phi}$ satisfies some $\psi_i$ and consequently no variable from $Y$ occurs in $\psi_i$. Hence, in the $i$th copy of $r_0[X]$ no variable in $Y$ is set to ``$-$'', and therefore $r_0[Y]\times r[i]\sub r[Yi]$. Since all tuples in $r_1[X]$ are joined with all index values of $i$, we also have $r_1[Y]\times r[i]\sub r[Yi]$. Hence, $Y\boto i$ is true in $r$ and $f_r(Y)=1$.

Assume then that $f_{\phi}(Y)=0$.  Then $f_{\phi}$ falsifies all $\psi_i$ implying that some variable from $Y$ occurs in every $\psi_i$. Therefore, every copy of $r_0[X]$ sets some variable in $Y$ to ``$-$''. This means that the values of $Y$ in $r_0[X]$ are not joined by all index values, thus rendering $Y\boto i$ false. Finally, since the size of $r$ is cubic in the number of clauses, the claim follows.
\end{proof}

That $\IAD$ is $\W{3}$-complete follows now from Lemmata \ref{inW3}-\ref{help2}.

\begin{theorem}
$\IAD$ is $\W{3}$-complete.
\end{theorem}

%\begin{comment}
Combining findings from this paper and \cite{Blasius0S16} we now know that the discovery problem for both ISs and INDs is $\W{3}$-complete, while for FDs  it is $\W{2}$-complete. Observe that these dependencies are either universal-existential (ISs and INDs) or universal  (FDs) first-order formulae. Reducing to propositional formulae this quantifier alternation translates to connective alternation at the outermost level, while the quantifier-free part is transformed at the innermost level to a conjunction of negative literals for ISs and INDs, or  a dual Horn formula (i,e., a disjunction with at most one negative literal) for FDs. Note that these data dependencies, among most others, generalize to either \emph{tuple-generating dependencies} (tgds) (i.e., first-order formulae of the form $\forall \tuple x(\phi(\tuple x) \rightarrow \exists \tuple y \psi(\tuple x,\tuple y))$ where $\phi$ and $\psi$ are conjunctions of relational atoms) or \emph{equality-generating dependencies} (egds) (i.e., first-order formulae of the form $\forall \tuple x(\phi(\tuple x) \rightarrow \psi(\tuple x))$ where $\phi$ is a conjunction of relational atoms and $\psi$ is an equality atom). It seems safe to conjecture that the reductions outlined above can be tailored to most subclasses of tgds (egds, resp.). Thus, likely no natural data dependency class with a reasonable notion of arity taken as the parameter has fixed-parameter complexity beyond $\W{3}$.
%\end{comment}

\section{Discovery algorithm}\label{s:alg}

Despite the computational barriers on generally efficient solutions to the independence discovery problem, we will now establish an algorithm that works efficiently on real-world data sets that exhibit ISs of modest arity. The algorithm uses level-wise candidate generation and, based on the downward-closure property of ISs, it terminates whenever we can find an arity on which the given data set does not exhibit any IS. The algorithm works therefore well within the computational bounds we have established: as the problem is \W{3}-complete in the arity, the discovery of any ISs with larger arities requires an exponential blow-up in discovery time.

%In this section we discuss an alternative IS discovery algorithm that uses level-wise candidate generation.
Henceforth, by the symmetry of the independence statement we identify each IS of the form $X\boto Y$ with the set $\{X,Y\}$.
As mentioned above, the arity of the IS $\{X,Y\}$ is given as $|XY|$% in this case the number of columns in the sets $X\cup Y$.
  The algorithm generates candidate ISs with increasing arity and tests whether each candidate ISs is valid. This algorithm is worst-case exponential in the number of columns. The relation from Table~\ref{t:ex} will be used to illustrate the algorithm.

\subsection{Disregard constant columns}

As our experiments illustrate later, the inclusion of columns that feature only one value (constant columns) renders the discovery process inefficient. A pre-step for our discovery algorithm is therefore the removal of all constant columns from the input data set.

\subsection{Generating candidate ISs}

The algorithm begins by generating all the possible ISs of arity two. Indeed, unary ISs need not be considered as they are satisfied trivially. On our running example the candidate ISs are \begin{center}\{\{a\},\{e\}\}, \{\{a\},\{re\}\}, \{\{a\},\{s\}\}, \{\{e\}, \{re\}\}, \{\{e\},\{s\}\}$, and $\{\{re\},\{s\}\}.\end{center}

\subsection{Validating candidates}

For the validation of the candidate ISs we exploit their characterization using Cartesian products. Indeed, for a candidate IS $X\bot Y$ to hold on $r$ it suffices that the equation $|r(X)|\times |r(Y)| = |r(XY)| $ holds.

As an example, consider the candidate IS set $\{\{e\},\{re\}\}$ where $X=\{e\}$ and $Y=\{re\}$. Both the \textit{education} and \textit{relationship} columns each feature two distinct values, that is, $|r(X)|=2=|r(Y)|$. The projection of the given table onto \{\emph{education}, \emph{relationship}\} features four distinct tuples:

\begin{center}
(bachelors, not-in-family), \\
(bachelors, husband), \\
(hs-grad, not-in-family), and \\
(hs-grad, husband),
\end{center}

so $|r(XY)|=4$. Given these values, $|r(X)|\times |r(Y)| = |r(XY)|$ holds and therefore $\textit{education} \bot\textit{relationship}$ does indeed hold on $r$.

Now consider the candidate IS set $\{\{e\},\{s\}\}$. Again, both the \textit{education} and \textit{sex} columns each have two distinct values, so $|r(X)|=2$ and $|r(Y)|=2$. However, the projection of the given table onto \{\emph{education}, \emph{sex}\} only features three tuples:

\begin{center}
(bachelors, male),\\
(hs-grad, male), and\\
(hs-grad, female).
\end{center}

In this case, $|r(X)|\times |r(Y)| = |r(XY)|$ is not satisfied and the IS \[\textit{education}\bot\textit{sex}\] does not hold.

\subsection{Computation of \MINCOVER}

Once all the candidate ISs of the current arity level have been validated, the actually valid ISs are added to \MINCOVER. For every new IS in \MINCOVER, we remove all ISs from \MINCOVER whose column sets are subsumed by the new IS. This step ensures that we end up with a \textit{minimal} quasi-cover. Candidate ISs whose validation failed are removed from the candidate set, but every validated IS is kept in the set of candidates to generate new candidates of the incremented arity in the next step.

In the case of our running example we will only have $\{\{e\}, \{re\}\}$ in \MINCOVER and \CANDIDATES.

\subsection{New candidate ISs}

If \CANDIDATES is either empty or there is some candidate of arity $|R|$, then the algorithm terminates. Otherwise, the remaining candidates of the current arity are used to create new candidates for validation at the next arity level. This is simply done by adding each remaining singleton to one attribute set of the remaining candidate IS.

In our running example, the algorithm generates the set of candidate ISs of arity three given the remaining candidate $\{\{e\}, \{re\}\}$. We obtain:
\begin{center}
$\{\{a,e\}, \{re\}\}$, $\{\{e,ra\}, \{re\}\}$, $\{\{s,e\}, \{re\}\}$, $\{\{e\}, \{a, re\}\}$, $\{\{e\}, \{ra,re\}\}$, $\{\{e\}, \{re,s\}\}$.
\end{center}

\subsection{Optimization before validation}

Before validation of the new candidates, however, we can further prune the new candidates. For this purpose, we use the downward-closure property of ISs. In fact, for $X\boto Y$ to hold on $r$, both $X\backslash\{A\}\boto Y$ and $X\boto Y\backslash \{B\}$ must hold on $r$ for all $A\in X$ and $B\in Y$. In particular, if $r$ does not satisfy any IS of arity $n$, then $r$ cannot satisfy any IS of arity $m>n$.

In our running example, none of the new candidates of arity three can possibly hold on $r$ since there is already a subsumed IS of arity two that does not hold on $r$. That is, the algorithm terminates on our running example at this point.

\subsection{Next arity level}

The algorithm starts a new iteration at the next arity level whenever there are any new candidates.

\subsection{Missing values}

Missing values occur frequently in data. While other solutions are possible, we adopt the view that data should speak for itself and neglect any tuples from consideration when it has a missing value on any of the attributes of the current candidate ISs under consideration.

\subsection{The algorithm}

Algorithm~\ref{alg} summarizes these idea into a bottom-up IS discovery algorithm that works in worst-case exponential time in the given number of columns.

\begin{algorithm}[t]
 \caption{Bottom-up Discovery Finder\label{alg}}
  \SetKwData{Left}{left}
  \SetKwData{Up}{up}
  \SetKwFunction{FindCompress}{FindCompress}
  \SetKwInOut{Input}{input}
  \SetKwInOut{Output}{output}

\Indm
  \Input{A relation $r$ over $R$.}
  \Output{The minimal quasi-cover $\MINCOVER$ of $\True(r)$}
\Indp
  \BlankLine

 $\CANDIDATES\gets \binom{R}{2}$\;
 $\NEWCANDIDATES \gets \emptyset$\;
 \While{$\CANDIDATES$ is non-empty}{
 	\For{$\{X,Y\}\in \CANDIDATES$}{
 		use a sorting algorithm to compute $|r(X)|,|r(Y)|,|r(XY)|$\;
 		%{\color{red}compute $|r(X)|\cdot |r(Y)|$ needed here?}\;
 		
 		\eIf{$|r(X)|\cdot |r(Y)|=|r(XY)|$ }
 			{add $\{X,Y\}$ to $\MINCOVER$\;
 			\ForAll{$A\in X$ and $Y\in B$}{ remove $\{X\setminus\{A\},Y\}$ and $\{X,Y\setminus \{B\}\}$ 				 from $\MINCOVER$}}
		{remove $\{X,Y\}$ from $\CANDIDATES$\;}	
	}
	\eIf{$\CANDIDATES$ is empty or has ISs of arity $|R|$}
	{stop}
	{\ForAll{$\{X,Y\}$ of arity $1$ greater than that of ISs in $\CANDIDATES$}{$m\gets 1$\;
		\ForAll{$A\in X$ and $B\in Y$}{
			\If{$\{X\setminus\{A\},Y\}$ or $\{X,Y\setminus \{B\}\}$ is not in \CANDIDATES}
			{$m\gets 0$}
			%{move to next section}
		}
		\If{$m=1$}
		{add $\{X,Y\}$ to $\NEWCANDIDATES$}
		%{move to next section}
	}}
	$\CANDIDATES \gets \NEWCANDIDATES$\;
    $\NEWCANDIDATES \gets \emptyset$
 }
\end{algorithm}

\begin{theorem}
Given a relation $r$ over schema $R$ with $k$ columns, Algorithm~\ref{alg} terminates in time $2^{O(k)}$ with the minimal quasi-cover of all ISs over $R$ that hold on $r$.
\end{theorem}

\begin{proof} The correctness of Algorithm~\ref{alg} is clear from the description above. Algorithm~\ref{alg} is worst-case exponential in the number $k$ of columns. Assume $r$ has $n$ rows and $k$ attributes. Step 5 takes time $O(n\log)$ per each sorting. This step is iterated at most $O(2^k)$ times during the computation. Hence in total Step 5 will take time $O(2^k n\log n)$.

In step 6, as there are at most $n$ rows, the input of this computation is of length $2\log n$. For example, schoolbook multiplication takes time $O(n^2)$ so we get an upper bound of $O((\log n)^2)$. Again, there is at most $O(2^k)$ iterations. Hence, we obtain an upper bound $O(2^k(\log n)^2)$.

Step 16 has at most $2^k$ iterations, step 18 at most $k^2$ iterations, each search through at most $2^k$ many candidates. Hence, the complexity is $2^{O(k)}$ at most.\qed
\end{proof}
 		
\section{Experiments}\label{s:experiments}

So far we have established the \NP-completeness and \W{3}-com\-plete\-ness of the discovery problem for independence statements, as well as the first algorithm that is worst-case exponential in the number of columns. As a means to illustrate practical relevance, we will apply our algorithm to several real-world data sets that have served as benchmarks for other popular classes of data dependencies.

\subsection{Data Sets}

As our algorithm is the first for the discovery of ISs, it is natural to apply it to data sets that serve as benchmark for discovery algorithms of other popular classes of data dependencies. The arguably most popular class are functional dependencies \cite{DBLP:journals/pvldb/PapenbrockEMNRZ15}, for
which the discovery problem is \W{2}-complete in the arity. Although it often makes little sense to compare results for one class to
another, the situation is a bit different here. In fact, independence and functional dependence are opposites in the sense that the
former require all combinations of left-hand side and right-hand side values, while the latter enforce unique right-hand side values for fixed left hand-side values. In total, we consider 17 data sets with various numbers of rows (from 150 to 250,000) and columns (from 5 columns to 223 columns). Some of the full data sets have been limited to just 1000 rows for functional dependencies already, and even
that proves to be a challenge for FD discovery algorithms and our algorithm. Note that all data sets are publicly available\footnote{\url{hpi.de/naumann/projects/repeatability/data-profiling/fds.html}}.

%\textit{Algorithm 2} was much faster and seems to be particular efficient for detecting ISs of small arity. All the benchmark dataset was run through \textit{Algorithm 2} although the algorithm runs out of memory before it can finish computing all the ISs for \textit{plista} and \textit{flight}. For pre-constant removal and pre-NA removal, only the results from datasets that are affected will be shown. For example, \textit{iris}  does not have any constant column nor NA values, so it's runtime and number of ISs would still be the same, and the result would not be repeated for ease of comparison. The overall results are shown in the table below, and smaller table broken down for constant and NA removal will be shown and discussed subsequently.

\subsection{Main Results}

Table~\ref{tab:results} shows our main results. We say that a column is \emph{constant} if only one value appears in it. We list for each data set the numbers of its columns (\#c), rows (\#r), constant columns (\#c-c), ISs in the quasi-minimal cover (\#IS) together with the maximum arity among any discovered IS (inside parentheses), FDs in a LHS-reduced cover (\#FD), and the best times to find the latter two numbers, respectively. In particular, IS time lists the runtime of Algorithm~\ref{alg} in seconds, and FD time is the fastest time to find \#FD by the algorithms in \cite{DBLP:journals/pvldb/PapenbrockEMNRZ15}.

\begin{table}[h]
	\centering\caption{Results of IS discovery and comparison}
	\label{tab:results}
		\begin{tabular}{l | r  r r| r r| r r }
			Data set 		& \#c & \#r & \#c-c & \#IS 		& \#FD	& IS time	& FD time	\\
		\hline
			iris 			& 5		& 150		& 0				& 0			& 4	& 0.02		& 0.1	\\
			balance	& 5		& 625		& 0		& 11 (4)  	& 1			& 0.19		& 0.1	\\
			chess			& 7		& 28,056	& 0		& 22 (4) 	& 1			& 4.36		& 1	\\
			abalone			& 9		& 4,177		& 0		& 0	 		& 137		& 0.37		& 0.6	\\
			nursery			& 9		& 12,960	& 0		& 127 (8) 	& 1			& 132.89	& 0.9	\\
			breast	& 11	& 699		& 0		& 0	 		& 46		& 0.047		& 0.5	\\
			bridges			& 13	& 108		& 0		& 0	 		& 142		& 0.03		& 0.2	\\
			echo	& 13	& 132		& 1		& 0	 		& 538		& 0.06		& 0.2		\\
			adult			& 14	& 48,842	& 0		& 9	(3) 	& 78		& 8.89		& 5.9	\\
			letter			& 17	& 20,000	& 0		& 0  		& 61		& 1.51		& 6.0	\\
			ncvoter			& 19	& 1,000		& 1 	& 0	 		& 758		& 0.18		& 1.1	\\
			hepatitis		& 20	& 155		& 0		& 21 (4)	& 8,250		& 15.45		& 0.8	\\
			horse			& 27	& 368		& 0		& 39 (3)	& 128,726	& 4.07		& 7.2	\\
			fd-red	& 30	& 250,000	& 0		& 0 		& 89,571 	& 186.02	& 41.1	\\
			plista			& 63	& 1,000 	& 24	& ?	(5)		& 178,152	& ML		& 26.9	\\
			flight			& 109	& 1,000 	& 39	& ?	(4)		& 982,631	& ML		& 216.5	\\
			uniprot			& 223	& 1,000 	& 22	& 0			& ?			& 16.61		& ML	\\
		\end{tabular}
\end{table}

Our first main observation is that the algorithm performs according to its design within the computational limits we have established for the problem. Indeed, the runtime of the algorithm blows up with the arity of the ISs that are being discovered as well as the columns that need to be processed. Larger row numbers means that the candidate ISs require longer validation times. In particular, the algorithm terminates quickly when no ISs can be discovered for small arities. The best example is \emph{uniprot} which does not exhibit any non-trivial ISs, and this can be verified within 16.61 seconds despite being given 223 columns (201 after removal of the columns that have a constant value). In sharp contrast, no FD discovery algorithm is known that can discover all FDs that hold on the same data set since there are too many.

During these main experiments we made several other observations that are worth further commenting on in subsequent subsections. Firstly, several of the real-world data sets contain constant columns, so the impact of those is worth investigating. Secondly, most of the data sets contain missing values, just like in practice. Hence, we would like to say something about the different ways of handling missing values. Finally, the semantic meaningfulness of the discovered ISs can be discussed. While the decision about the meaningfulness will always require a domain expert, a ranking of the discovered ISs appears to be beneficial in practice and interesting in theory.

\subsection{Constant Columns}

Out of the 17 test data sets, five have columns with constant values. Table~\ref{t:const} compares runtime and results when constant columns are kept and removed, respectively, from the input data set.

\begin{table}[h]
	\centering\caption{Illustrating the effect of constant column removal}
		\label{t:const}%\footnotesize
		\begin{tabular}{l | r  r r |r|r r r }
			&  \multicolumn{3}{c|}{Before removal} & & \multicolumn{3}{c}{After removal}\\
			Data set 		& \#c & time 	& \#IS 	& \#c-c & \#nc-c & time 	& \#ISs\\
			\hline
			echo        	& 13	& 2082.5	& 1(13)	& 1 		& 12	& 	0.06	& 0	 \\
			ncvoter			& 19	& TL		& 1(19)	& 1			& 18 	&	0.21	& 0	 \\
			plista			& 63	& ML 		& ? (4) & 24 		& 39 	&	ML		& ? (5)\\
			flight			& 109	& ML		& ? (4) & 39 		& 70  	&	ML		& ? (4)\\
			uniprot			& 223	& TL		& 1(223)& 22 		& 201	&  	16.61	& 0	\\
		\end{tabular}
\end{table}

Indeed, the table provides clear evidence that the removal of constant columns from the input data set to independence discovery algorithms is necessary. Evidently, every data set over schema $R$ trivially satisfies the IS \[A\boto R-A\] whenever $A$ is a constant column. Hence, our algorithm would need to explore all arity levels until all columns of the schema are covered. This, however, is prohibitively expensive as illustrated in Table~\ref{t:const}.

\subsection{Missing Values}

Out of the 17 data sets nine have missing values. The number of these missing values is indicated in the table below. There are different ways in which missing values can be handled, and we have investigated two basic approaches to illustrate differences. In the first approach, we simply treat occurrences of missing values as any other value. In the second approach, we do not consider tuples whenever they feature a missing value in some column that occurs in a candidate IS. For the purpose of stating the precise semantics, we denote an occurrence of a missing value by the special marker symbol ``NA" and distinguish the revised IS from the previous semantics by writing $\bot_0$ instead of $\bot$. An IS $X\bot_0 Y$ is satisfied by $r$ if and only if for every two tuples $t_1,t_2\in r$ such that $t_1(A)\not=$ ``NA" and $t_2(A)\not= $ ``NA" for all $A\in X\cup Y$, there is some $t\in r$ such that $t(X)=t_1(X)$ and $t(Y)=t_2(Y)$.

\begin{table}[h]
\centering		\caption{Differences of semantics for missing values}
		\label{tab:alg2na}%\footnotesize
		\begin{tabular}{l | r  r r |r r| r r}
			&&&&  \multicolumn{2}{|c|}{$\bot$}  & \multicolumn{2}{c}{$\bot_0$}\\			
			dataset 		& \#c & \#r 	& \#miss & time 		& \#IS 		& time		&\#IS\\
			\hline
			breast 	        & 11	& 699		& 16			& 0.04		& 0	 		& 0.04		& 0	\\
			bridges			& 13	& 108		& 77			& 0.03		& 0	 		& 0.18		& 4(3)\\
			echo	        & 13	& 132		& 	132		& 0.06		& 0	 		& 0.38		& 5(4)\\
			ncvoter			& 19	& 1,000		& 2863		    & 0.18		& 0	 		& 0.18 			& 0	\\
			hepatitis		& 20	& 155		& 167			& 15.45		& 21 (4)	& 1597.23	& 855 (6) \\
			horse			& 27	& 368		& 1605			& 4.07		& 39 (3)	& 35.56		& 112 (3)	\\
			plista			& 63	& 1,000 	&	23317        & ML		& ? (5)		& ML		& ?	(5)\\
			flight			& 109	& 1,000 	&	51938        & ML		& ?	(4)		& ML		& ?	(4)\\
			uniprot			& 223	& 1,000 	&  179129      	& 16.61		& 0			& 16.61			& 0	\\
		\end{tabular}
\end{table}

The main message is that the choice of semantics for missing values clearly affects the output and runtime of discovery algorithms. It is therefore important for the users of these algorithms to make the right choice for the applications they have in mind. In general it is difficult to pick a particular interpretation for missing values, and relying just on the non-missing values might be the most robust approach under different interpretations.

\subsection{Redundant independence statements}

It is possible that some of the IS atoms that our algorithm discovers are already implied by other ISs that have been discovered. Here, a set $\Sigma$ of ISs is said to imply an IS $\varphi$ if and only if every relation that satisfies all ISs in $\Sigma$ will also satisfy $\varphi$. Hence, if $\Sigma$ does imply $\varphi$, then it is not necessary to list $\varphi$ explicitly. In other words, if we list all ISs in $\Sigma$ it would be redundant to list $\varphi$ as well. The question arises how we can decide whether $\Sigma$ implies $\varphi$. Table~\ref{tab-rules} shows an axiomatic characterization of the implication problem for ISs \cite{DBLP:conf/wollic/KontinenLV13}.

\begin{table}
\caption{Axiomatization $\mathfrak{I}$ of Independence in Database Relations\label{tab-rules}}
\[\fbox{$\begin{array}{c@{\hspace*{.25cm}}c}
\cfrac{}{X\boto\emptyset} & \cfrac{X\boto Y}{Y\boto X} \\
\text{(trivial independence, $\mathcal{T}$)} & \text{(symmetry, $\mathcal{S}$)} \\ \\

\cfrac{X\boto YZ}{X\boto Y} & \cfrac{X\boto Y\quad XY\boto Z}{X\boto YZ} \\
\text{(decomposition, $\mathcal{D}$)} & \text{(exchange, $\mathcal{E}$)}
\end{array}$}\]
\end{table}

As an illustration, we mention a valid IS of arity 8 from the \textit{nursery} data set. Indeed, the IS \[\varphi=\{1\}\bot \{2,3,4,5,6,7,8\}\] is implied by the following four ISs:
\begin{center}
$\{1,5,6,7\}\bot\{2,3,4,8\}$, $\{1,5,6,8\}\bot\{2,3,4,7\}$, \\
$\{1,5,7,8\}\bot\{2,3,4,6\}$, and $\{1,6,7,8\}\bot\{2,3,4,5\}$,
\end{center}
using repetitive application of the decomposition and exchange rule. Therefore, \[ \{1\}\bot\{2,3,4,5,6,7,8\}\] is redundant and can be removed without loss of information. This illustrates that there is still scope to reduce the set of discovered ISs without loss of information.

\subsection{Semantical Meaningfulness}

Algorithms cannot determine if a discovered atom is semantically meaningful for the given application domain, or only holds accidentally on the given data set. Ultimately, this decision requires a domain expert. For example, the output for the \textit{adult} data set includes $\textit{education}\bot\textit{sex}$ and $\textit{relationship}\bot\textit{sex}$. Considering the \textit{adult} data set was collected in 1994 in the United States, the IS $\textit{education}\bot\textit{sex}$ makes sense semantically. However, $\textit{relationship}\bot\textit{sex}$ is not really semantically meaningful when the \textit{relationship} column contains values such as \textit{husband} and \textit{wife}. Same sex marriage was not legalized in the United States until 2015, so if a person has the value \textit{husband} for \textit{relationship}, then he should also have the value \textit{male} for \textit{sex}, and similarly for \textit{wife} and \textit{female}. However, there is one tuple in the data set which contains the values \textit{husband} and \textit{female} and three tuples with the combination \textit{wife} and \textit{male}, which is why the IS was validated. With 48,842 tuples in total and only four with these value combinations, it is likely that the entries are mistakes and $\textit{relationship}\bot\textit{sex}$ really should not have been satisfied.

\subsection{Arity, Column, and Row Efficiency}

It is further interesting to illustrate impacts on the runtime and number of candidate ISs with growing numbers of arity, columns, and rows. As an example, we have conducted such experiments on the data set \emph{nursery}.

\begin{figure}[h]
	\begin{center}
		\includegraphics[width=0.5\columnwidth]{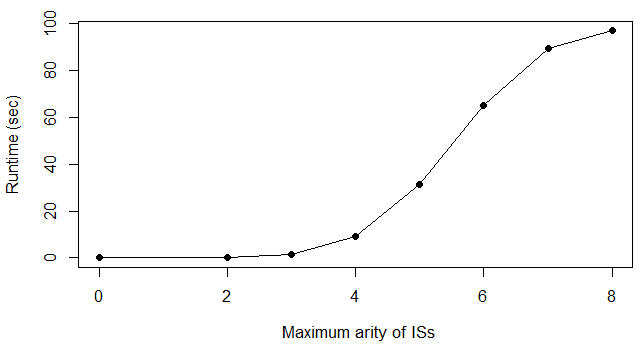}\includegraphics[width=0.5\columnwidth]{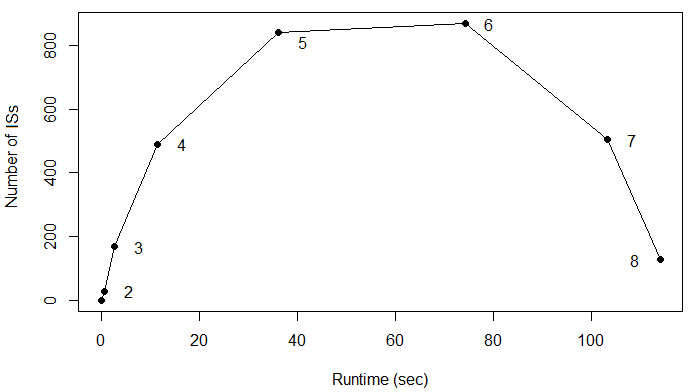}\\
		\caption{Runtime and valid IS numbers under arity increases}
		\label{fig:nursearity}
	\end{center}
\end{figure}

Figure \ref{fig:nursearity} shows unsurprisingly that the runtime increases as the arity of the candidate ISs increases. Indeed, the algorithm first tests all the ISs of lower arity. The rate of increase is reasonably flat at the start, then quite drastic, and then  reasonably flat again. This is correlated to the number of ISs at the current arity level, since the more ISs there are the more candidate ISs the algorithm needs to validate. As arity increases, multiple valid ISs may be covered by fewer IS with higher arity, thus decreasing the IS count. For example, one IS of arity five on \textit{nursery} is $\{1\}\bot\{2,3,4,5\}$, covering
\begin{center}
$\{1\}\bot\{2,3,4\}$, $\{1\}\bot\{2,3,5\}$, $\{1\}\bot\{2,4,5\}$, and $\{1\}\bot\{3,4,5\}$
\end{center}
from the arity level 4.

 \begin{figure}[H]
 	\begin{center}
	 	\includegraphics[width=0.5\columnwidth]{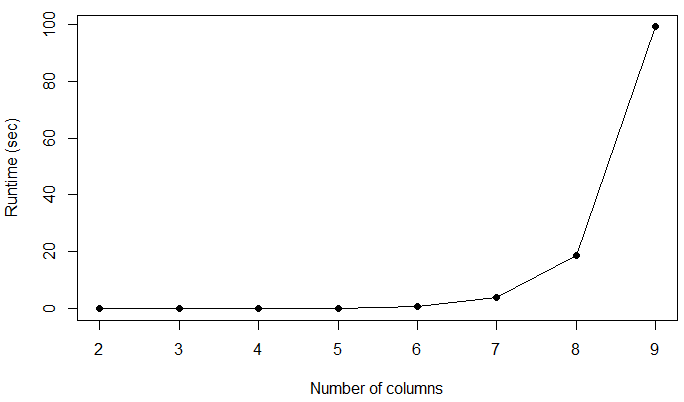}\includegraphics[width=0.5\columnwidth]{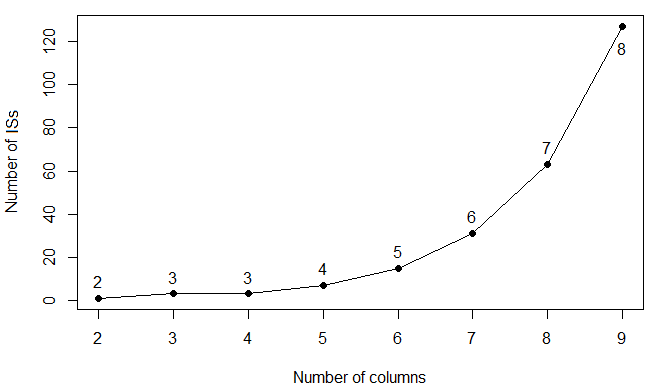}\\
	 \end{center}
	 \caption{Runtime and valid IS numbers under column increases}
	 \label{fig:nursecol}
 \end{figure}

For experiments with column efficiency we run our algorithm on projections of \emph{nursery} on the following randomly created subsets of columns: \begin{center} $\{7,8\}$, $\{7,8,9\}$, $\{7,8,9,3\}$, $\{7,8,9,3,1\}$, $\{7,8,9,3,1,6\}$, $\{7,8,9,3,1,6,4\}$,  $\{7,8,9,3,1,6,4,2\}$, and $\{7,8,9,3,1,6,4,2,5\}$.\end{center} Figure~\ref{fig:nursecol} shows an exponential blow up of the runtime and number of valid ISs in the growing number of columns.

\begin{figure}[h]
	\begin{center}
		\includegraphics[width=0.5\columnwidth]{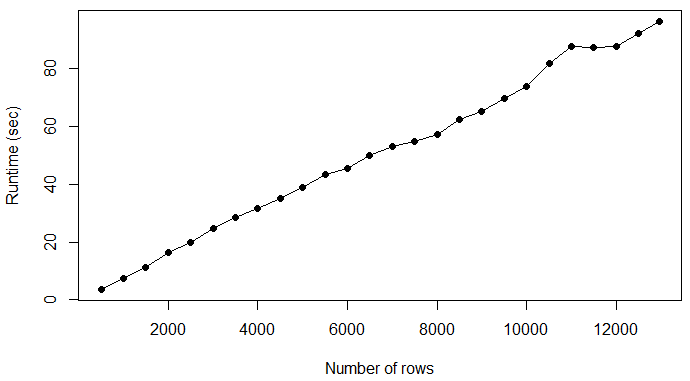}\includegraphics[width=.5\columnwidth]{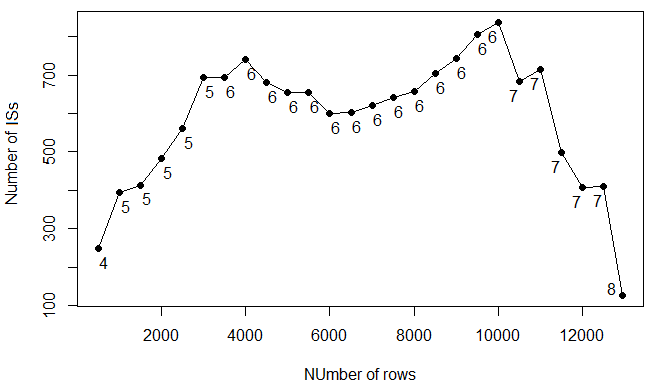}\\
	\end{center}
	\caption{Runtime and valid IS numbers under row increases}
	\label{fig:roweff}
\end{figure}

For experiments with row efficiency we run our algorithm on 26 randomly created subsets of \emph{nursery}. Starting with 500 rows, the next data set is created from the previous one by adding 500 additional randomly selected rows from the remaining data set. As can be seen in Figure~\ref{fig:roweff}, the runtime grows linearly in the growing number of rows. This is because an increase in rows does i) not affect the number of candidate ISs the algorithm needs to validate, but ii) only slows down the validation process of candidate ISs.
	
Figure~\ref{fig:roweff} shows a steady increase in the number of ISs as the number of rows increases, but a decline as the rows approximate the full data set. Indeed, with more rows, the subsets are more likely to contain all the combinations of values to satisfy independence. Moreover, ISs of higher arity cover multiple ISs of lower arity, and this coverage eventually results in a decline for the number of valid ISs.

\section{Approximate Independence}\label{sec:approximate}

In practice, independence is a strong assumption. Even in cases where an independence statement should hold, it may not hold because of data quality or other problems. For many important tasks it is not necessary that an independence statement holds, but it is more useful to know to which degree the IS holds. This can be formalized by the notion of approximate independence. We will use this section to formally introduce this notion, and conduct experiments on our benchmark data to illustrate how the degree of independence affects the number of discovered approximate independence statements as well as the runtime efficiency of the algorithm that discovers them. We conclude this section with some examples that provide some qualitative analysis of approximate independence statements in our benchmark data.

\subsection{Introducing Approximate Independence}

The intuition of an approximate independence statement is as follows. We know that a relation $r$ satisfies the IS $X\bot Y$ if and only if $|r(XY)|=|r(X)|\times |r(Y)|$. In fact, as $|r(XY)|\le |r(X)|\times |r(Y)|$ is always satisfied, $r$ satisfies $X\bot Y$ if and only if $|r(XY)|\ge |r(X)|\times |r(Y)|$ holds. Hence, the ratio
\[ \epsilon_r=\cfrac{|r(XY)|}{|r(X)|\times |r(Y)|} \]
quantifies the degree by which an IS holds on a given data set. If the ratio is 1, then the IS holds. Consequently, we can relax the IS assumption according to our needs by stipulating that the ratio $\epsilon_r$ is not smaller than some threshold $\epsilon\in [0,1]$.

\begin{definition}
For a relation schema $R$, subsets $X,Y\subseteq R$, and a real $\epsilon\in [0,1]$, we call the statement $X\bot_\epsilon Y$ an \emph{approximate independence statement} (aIS). The aIS  $X\bot_\epsilon Y$ is said to hold on a relation $r$ over $R$ if and only if \[\epsilon_r=\cfrac{|r(XY)|}{|r(X)|\times |r(Y)|}\ge\epsilon\] holds. We call $\epsilon_r$ the independence ratio of $X\bot Y$ in $r$.
\end{definition}

During cardinality estimation in query planning it is prohibitively expensive to compute the cardinality of projections, which means that independence between (sets of) attributes is often simply assumed to ease computation. Hence, cardinalities are only estimated. However, knowing what the independence ratio of a given IS in a given data set is, means that we can replace the estimation by a precise computation. Consequently, we would like to know which approximate ISs hold on a given data set. This, however, is a task of data profiling, and we can adapt Algorithm~\ref{alg} to compute all aISs for a given threshold $\epsilon$. Instead of checking in line 6 whether $|r(X)|\cdot |r(Y)|=|r(XY)|$ holds, we simply check whether $\epsilon_r\ge\epsilon$ holds.

The computational limitations of IS discovery carry over to the approximate case. The approximate variant of $\IAD$ asks to determine whether $r$ satisfies some $k$-ary aIS $X\bot_\epsilon Y$ for a given relation $r$, a natural number $k$, and a threshold $\epsilon\in[0,1]$. An immediate consequence of our earlier results is that this problem is $\NP$-hard and $\W{3}$-hard in the arity since setting $\epsilon =1$ brings us back to $\IAD$. Whether some reduction works vice versa is not so clear. However, that this problem is in $\NP$ follows again by a simple argument; the only difference is, as stated above, that one now has to verify an inequality statement instead of an equality statement.

\subsection{Experiments}

With this simple adaptation of Algorithm~\ref{alg}, we conducted additional experiments to discover all aIAs whose independence ratio in a given benchmark data set meet a given threshold $\epsilon$.

Firstly, the experiments demonstrate that the number of aISs typically increases with lower thresholds. This is not surprising since the aISs that hold with a threshold $\epsilon$ will also hold with a threshold $\epsilon'\le\epsilon$. However, if new aISs are added for lower thresholds, then these may capture multiple aISs, typically but not exclusively when the new aISs have higher arity. In such cases the actual number in the representation of the output can be lower than that for bigger thresholds. In all of the subsequent figures, the left-hand side shows the different numbers of aISs for different choices of the threshold $\epsilon$. The data labels on these figures refer to the maximum arity across all of the aISs that have been discovered for the given threshold.

Secondly, the experiments demonstrate that the runtime of the algorithm increases when the threshold $\epsilon$ is lowered. Again, this is not surprising because of the increasing numbers of candidate and valid aISs that occur with lower thresholds. The right-hand side of all the subsequent figures illustrates the runtime behavior for different choices of the thresholds. As before, there are also cases in which the runtime becomes faster with lower thresholds. This typically occurs when new aISs are found (quickly) that cover many other valid aISs.

 \begin{figure}[H]
 	\begin{center}
	 	\includegraphics[width=0.5\columnwidth]{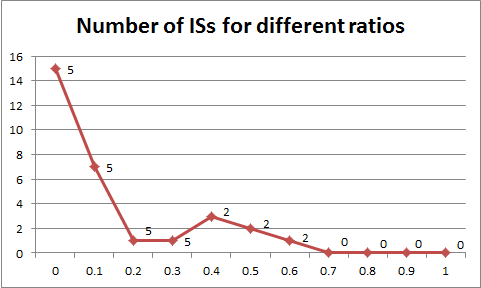}\includegraphics[width=0.5\columnwidth]{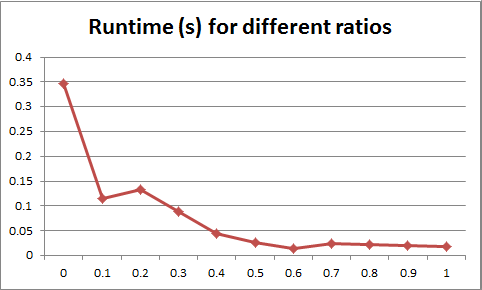}\\
	 \end{center}
	 \caption{Iris: IS numbers and runtime under different ratios}
	 \label{fig:iris-ratio}
 \end{figure}

 \begin{figure}[H]
 	\begin{center}
	 	\includegraphics[width=0.5\columnwidth]{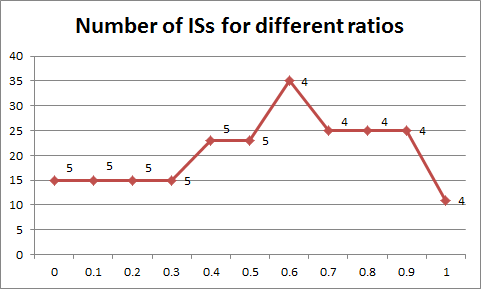}\includegraphics[width=0.5\columnwidth]{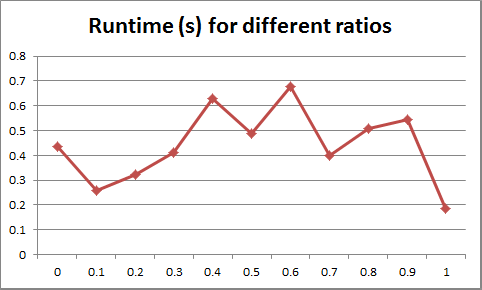}\\
	 \end{center}
	 \caption{Balance: IS numbers and runtime under different ratios}
	 \label{fig:balance-ratio}
 \end{figure}

 \begin{figure}[H]
 	\begin{center}
	 	\includegraphics[width=0.5\columnwidth]{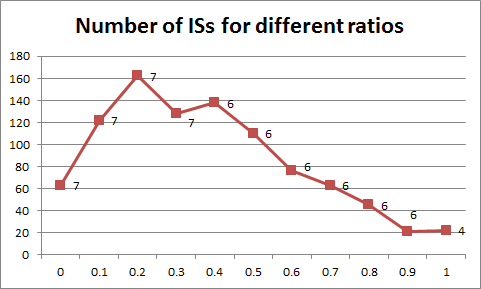}\includegraphics[width=0.5\columnwidth]{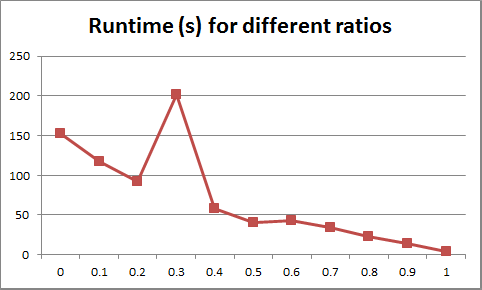}\\
	 \end{center}
	 \caption{Chess: IS numbers and runtime under different ratios}
	 \label{fig:chess-ratio}
 \end{figure}

 \begin{figure}[H]
 	\begin{center}
	 	\includegraphics[width=0.5\columnwidth]{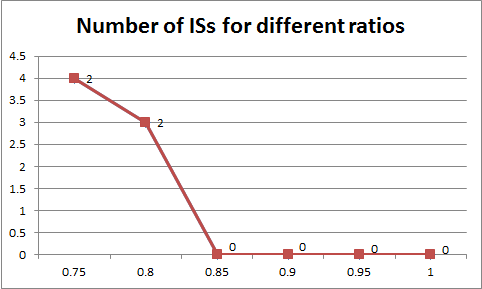}\includegraphics[width=0.5\columnwidth]{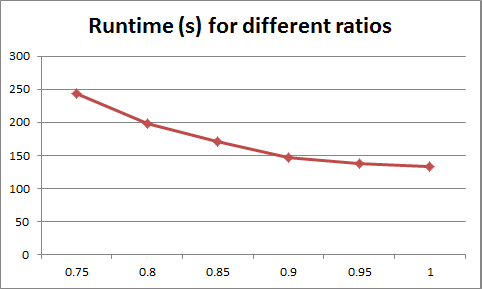}\\
	 \end{center}
	 \caption{Abalone: IS numbers and runtime under different ratios}
	 \label{fig:abalone-ratio}
 \end{figure}

 \begin{figure}[H]
 	\begin{center}
	 	\includegraphics[width=0.5\columnwidth]{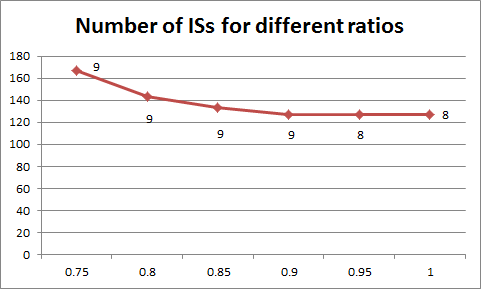}\includegraphics[width=0.5\columnwidth]{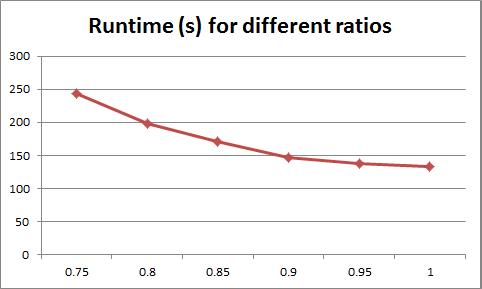}\\
	 \end{center}
	 \caption{Nursery: IS numbers and runtime under different ratios}
	 \label{fig:nursery-ratio}
 \end{figure}

 \begin{figure}[H]
 	\begin{center}
	 	\includegraphics[width=0.5\columnwidth]{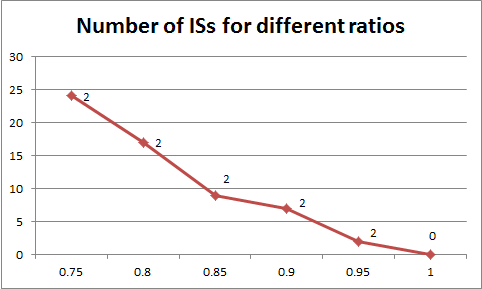}\includegraphics[width=0.5\columnwidth]{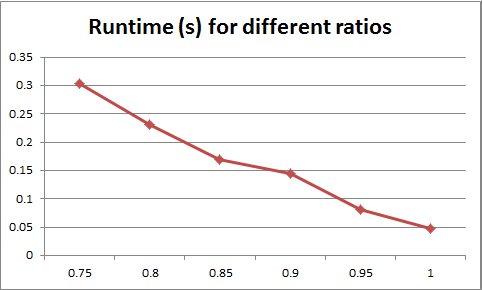}\\
	 \end{center}
	 \caption{Breast: IS numbers and runtime under different ratios}
	 \label{fig:breastcancer-ratio}
 \end{figure}

 \begin{figure}[H]
 	\begin{center}
	 	\includegraphics[width=0.5\columnwidth]{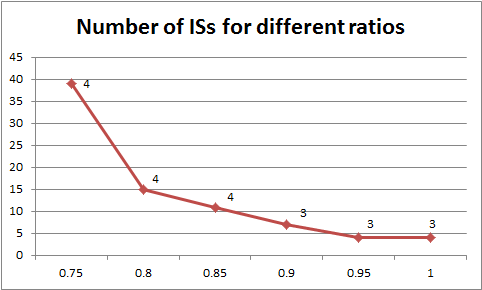}\includegraphics[width=0.5\columnwidth]{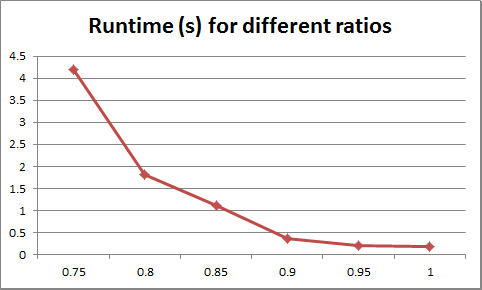}\\
	 \end{center}
	 \caption{Bridges: IS numbers and runtime under different ratios}
	 \label{fig:bridges-ratio}
 \end{figure}

 \begin{figure}[H]
 	\begin{center}
	 	\includegraphics[width=0.5\columnwidth]{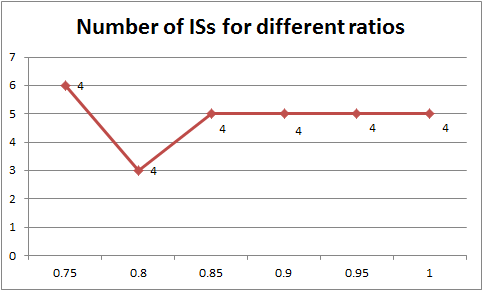}\includegraphics[width=0.5\columnwidth]{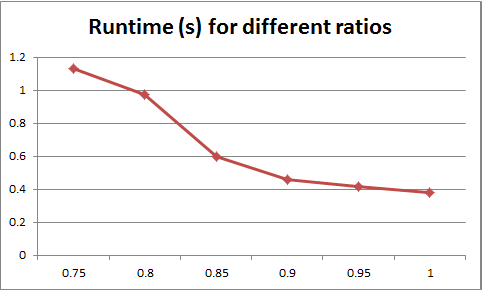}\\
	 \end{center}
	 \caption{Echo: IS numbers and runtime under different ratios}
	 \label{fig:echo-ratio}
 \end{figure}

 \begin{figure}[H]
 	\begin{center}
	 	\includegraphics[width=0.5\columnwidth]{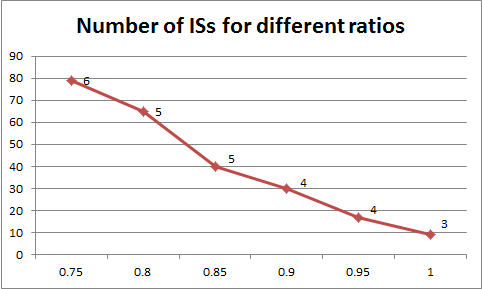}\includegraphics[width=0.5\columnwidth]{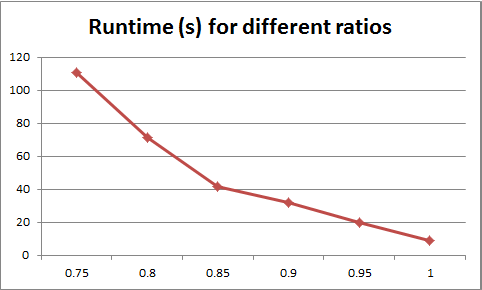}\\
	 \end{center}
	 \caption{Adult: IS numbers and runtime under different ratios}
	 \label{fig:adult-ratio}
 \end{figure}

 \begin{figure}[H]
 	\begin{center}
	 	\includegraphics[width=0.5\columnwidth]{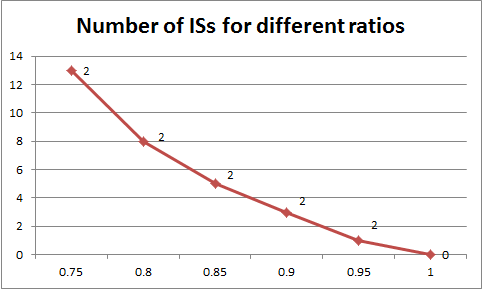}\includegraphics[width=0.5\columnwidth]{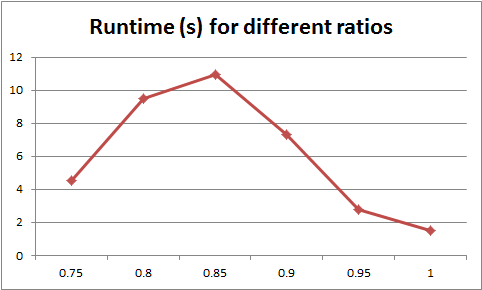}\\
	 \end{center}
	 \caption{Letter: IS numbers and runtime under different ratios}
	 \label{fig:letter-ratio}
 \end{figure}

 \begin{figure}[H]
 	\begin{center}
	 	\includegraphics[width=0.5\columnwidth]{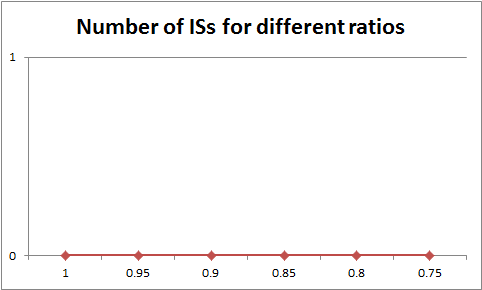}\includegraphics[width=0.5\columnwidth]{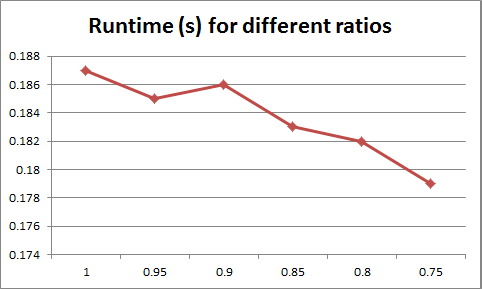}\\
	 \end{center}
	 \caption{NCVoter: IS numbers and runtime under different ratios}
	 \label{fig:ncvoter-ratio}
 \end{figure}

\subsection{Some qualitative analysis}

One motivation for approximate independence statements is their ability to recall actual independence statements that are not satisfied on the given data set due to some dirty data. In fact, Algorithm~\ref{alg} for the discovery of ISs can only discover ISs that hold on the data set, so even when there are ISs that should but do not hold, then the unmodified algorithm cannot discover them. However, after modifying Algorithm~\ref{alg} to discover all aISs for some given threshold $\epsilon$, some ISs that should actually hold can be discovered.

One example occurs in the data set \emph{hepatitis}. For the data set to be representative the columns \emph{age} and \emph{sex} should really be independent, but they are not. In fact, there are 49 distinct values for \emph{age} in the data set, and 2 distinct values for \emph{sex} in the data set, but 60 distinct value combinations on the projection onto $\{\emph{age},\emph{sex}\}$, so $49\cdot 2=98\not=60$. Hence, the IS $\emph{age}\bot\emph{sex}$ became only discoverable after choosing $\epsilon=0.6$, because its independence ratio in \emph{hepatitis} is $\approx 0.612$.

Ultimately, only a domain expert can make the decision whether an (approximate) IS is or should actually be valid. However, without looking at aISs, domain experts may never be guided towards considering an aIS that might be valid. On the other hand, the smaller the threshold the larger the number of aISs to consider. So, ultimately the choice of the threshold $\epsilon$ is important, too. For example, in the data set \emph{nursery} the ISs $6\bot 11$ and $9\bot 11$ have independence ratio $\epsilon_{\textit{nursery}}=0.95$, so are not discoverable as ISs. Column 11 indicates whether a cancer is \emph{benign} or \emph{malign}, while columns 6 and 9 indicate a \emph{single epithelial cell size} and \emph{normal nucleoli}, so could potentially be actual ISs, but it requires some domain expertise whether that is the case.

Apart from these considerations, however, approximate ISs have other uses as the approximation ratio $\epsilon_r$ is important for other tasks, such as cardinality estimation in query planning.

\section{Related Work}\label{sec:related}

Our research is the first to study the discovery problem for the simplest notion of independence. This is surprising for several reasons: 1) notions of independence are essential in many areas, in particular databases, artificial intelligence, and computer science, 2) the discovery problem has been an important computational problem for decades, and has recently gained new popularity in the context of artificial intelligence (under the name learning), big data, data mining, and data science, and 3) data profiling is an area of interest for researchers and practitioners, and while the discovery problem of many database constraints has been extensively studied in the past, this has not been the case for independence statements.
Indeed, in the context of databases different popular notions of independence have been studied. The concept of ISs that we use here has been studied as early as 1980 by Jan Paredaens \cite{DBLP:journals/jcss/Paredaens80}. He used the name \emph{crosses} instead, most likely as a reference to one of the most fundamental query operators: the cross product. The axiomatization from Table~\ref{tab-rules} is essentially the same as established in \cite{DBLP:journals/jcss/Paredaens80}, except that the attribute sets in crosses are defined to be disjoint while they do not need to be disjoint in ISs. As a foundation for distributed computing, graphical reasoning, and Bayesian networks, Geiger, Pearl, and Paz axiomatized so-called pure independence statements \cite{geiger:1991}, which are the probabilistic variants of ISs. Notably, the axioms for probabilistic pure independence are very similar to those for crosses. In the context of database schema design, multivalued dependencies (MVDs) were introduced by Ronald Fagin as an expressive class of data dependencies that cause a majority of redundant data values. Indeed, a relation satisfied an MVD if and only if the relation is the lossless join between two of its projections \cite{fagin77}. This fundamental decomposition property serves as a foundation for the well-known Fourth Normal Form (4NF) \cite{fagin77}. Embedded multivalued dependencies (EMVDs) are MVDs that hold on a projection of given relation, and are therefore even more expressive. Unfortunately, the finite and unrestricted implication problems are undecidable \cite{herrmann:2006,Herrmann2006}. In the context of artificial intelligence and statistics, the concept of an MVD is equivalent to the concept of saturated conditional independence. However, the more general and more important concept of conditional independence is not equivalent to the concept of embedded MVDs \cite{studeny:1993}, and the decidability of the implication problem for conditional independence is still open. The duality/similarity between concepts of independence continue even further. In data cleaning, the concept of conditional dependencies were introduced recently \cite{Fan:2008:CFD}. In this context, the word conditional refers to the fact that the dependency must not necessarily hold for all values of the involved attributes, but only conditional on specific values. In AI, this extension is known as \emph{context-specific independencies} \cite{Boutilier:1996}.

Much attention has been devoted to discovering conditional independencies in AI. The task of learning Bayesian networks to encode the underlying dependence structure of data sets is NP-complete and has been the topic of numerous articles and books \cite{Chickering1996,Heckerman1995,Neapolitan:2003}. Many algorithms employ so-called independence-based approach in which conditional independence tests are performed on the data sets and successively used to constrain the search space for the underlying graphical structure (e.g., the PC and SGS algorithms for Bayesian networks \cite{Spirtes2000}, or the GSMN algorithm for Markov networks %\cite{Margaritis:1999}).
\cite{Bromberg:2009}). Vice versa, assuming that the Bayesian network is given, the method of d-separation provides a tool for tractable identification of conditional independencies between random variables \cite{pearl90}.

Despite the fact that the discovery algorithms for various popular classes of data dependencies perform well in practice, there are usually no theoretical performance guarantees. This is not very surprising as all three problems are known to be likely intractable: finding a minimum unique column combination is $\NP$-complete \cite{DBLP:journals/jacm/BeeriDFS84} and cannot be approximated within a factor of $1/4\log{n}$ (under reasonable complexity assumptions) \cite{DBLP:conf/cocoon/AkutsuB96}, finding a minimum functional dependency is also $\NP$-complete \cite{davies:1994} and finding a maximum inclusion dependency is $\NP$-complete even for restricted cases \cite{DBLP:journals/ijis/KantolaMRS92}. The parameterized complexity for the discovery of unique column combinations, functional and inclusion dependencies was recently studied in \cite{Blasius0S16}. Parameterized on the arity, that is, the size of $X$ for all, a unique column combination $\textit{u}(X)$, a functional dependency (FD) $X\to A$, and an inclusion dependency (IND) $r[X]\sub r'[Y]$, it was shown that the discovery problems are complete for the second and third levels of the $\Wsym$ hierarchy, respectively. The case for inclusion dependencies is particularly interesting as many natural fixed-parameter problems usually belong to either $\W{1}$ or $\W{2}$. Our results about the $\W{3}$-completeness of the discovery problem for independence statements in their arity is therefore completing the picture by another interesting class.

In data profiling, for a recent survey see \cite{DBLP:series/synthesis/2018Abedjan}, investigations on the discovery problem have mostly been targeted at unique column combination \cite{DBLP:journals/pvldb/HeiseQAJN13,DBLP:journals/vldb/KohlerLLZ16,DBLP:conf/vldb/SismanisBHR06,DBLP:journals/pvldb/WeiL19-2}, functional dependencies \cite{DBLP:conf/sigmod/PapenbrockN16,DBLP:conf/icde/WeiL19,DBLP:journals/pvldb/WeiL19}, and inclusion dependencies \cite{DBLP:journals/tods/TschirschnitzPN17}, and due to the rise of data quality problems also on their conditional/context-specific variants such as conditional functional dependencies \cite{DBLP:journals/tkde/FanGLX11} and conditional inclusion dependencies \cite{DBLP:conf/cikm/BauckmannALMN12}. In contrast, notions of independence have only received restricted attention in data profiling. In fact, only MVDs have been considered so far and only by few authors \cite{DBLP:journals/ida/SavnikF00}. Neither independence statements, nor their approximate nor their context-specific variants have been explored in terms of the discovery problem. Our article closes this gap, and hopes to initiate research on the discovery problem for more sophisticated notions of independence.

 \begin{figure}[t]
 	\begin{center}
	 	\includegraphics[width=\columnwidth]{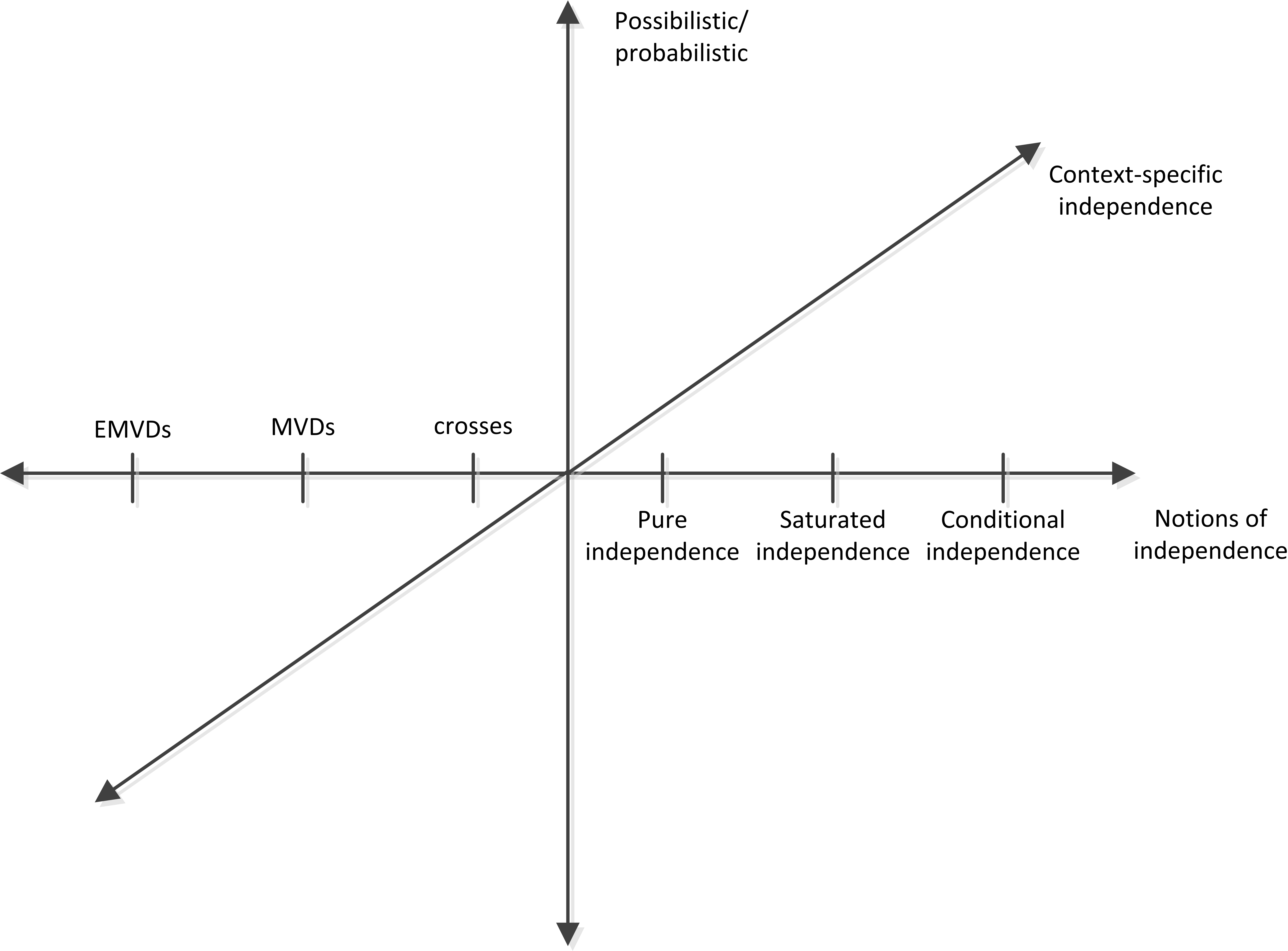}\\
	 \end{center}
	 \caption{Dimensions of notions for independence}
	 \label{fig:future}
 \end{figure}

\section{Future Work}\label{sec:future}

For future work we encourage research on the discovery of other notions of independence, their context-specific variants, their uncertain variants and the combination of those. Figure~\ref{fig:future} illustrates the dimensions that lead to complex notions of independence that can be explored. For example, embedded multivalued dependencies (EMVDs) are an
expressive class of data dependencies. Knowledge about which EMVDs hold on a given relational database would provide various options for query optimization. Con\-text\--specific variants form an orthogonal dimension, which refer to the specialization for a given notion of independence in the sense that the statement is not necessarily satisfied for all values of an attribute, but maybe only for specified fixed values on those attributes. For example, approximate context-specific ISs would be very helpful for cardinality estimation in query planning. Yet another orthogonal dimension can be considered by different choices of a data model. While we have limited our exposition to the relational model of data, other interesting data models include Web models such as JSON, RDF, or XML, or uncertain data models such as probabilistic and possibilistic data models. Of course, in artificial intelligence, machine learning, and statistics, the concepts of pure, saturated, and conditional independence are fundamental, specifically for distributed computations and graphical models.

\section{Conclusion}\label{s:conclusion}

We have initiated research on the discovery of independence concepts from data. As a starting point, we investigated the problem to compute the set of all independence statements that hold on a given data set. We showed that the decision variant of this problem is \NP-complete and \W{3}-complete in the arity. Under these fundamental limitations of general tractability, we designed an algorithm that discovers valid independence statements of incrementing arity. Once no valid statements can be found for a given arity, we are assured that no more valid statements exist. The behavior of the algorithm has been illustrated on various real-world benchmark data sets, showing that valid statements of low arity can be found efficiently on larger data sets, while identifying valid statements of higher arity is costly as expected from the hardness results of the problem. We have further illustrated how to adapt our algorithm to the new notion of an approximate independence statement, which only needs to hold with a given threshold. Approximate independence statements indicate with which ratio an independence statement holds on a given data set, which is useful knowledge for many applications such as cardinality estimation in query planning. We have outlined various directions of future research with more advanced notions of independence that have huge application potential in relational and probabilistic databases, but also for graphical models in artificial intelligence.

\bibliographystyle{spmpsci}      % mathematics and physical sciences
\bibliography{biblio}

\end{document}